\documentclass[lettersize,journal]{IEEEtran}
\usepackage{amsmath,amsfonts}
\usepackage{algorithmic}
\usepackage{algorithm}
\usepackage{array}
\usepackage[caption=false,font=normalsize,labelfont=sf,textfont=sf]{subfig}
\usepackage{textcomp}
\usepackage{stfloats}
\usepackage{url}
\usepackage{verbatim}
\usepackage{graphicx}
\usepackage{cite}

\usepackage{xcolor}
\usepackage{soul}

\usepackage{amsthm}

\newtheorem{definition}{Definition}
\newtheorem{proposition}{Proposition}
\newtheorem{theorem}{Theorem}

\usepackage{siunitx}
\begin{document}

\title{Channel Chart Location Privacy Based on Geo-Indistinguishability}

\author{Atsu Kokuvi Angélo Passah, Rodrigo C. de Lamare, \IEEEmembership{Fellow,~IEEE,} and Arsenia Chorti, \IEEEmembership{Senior Member,~IEEE} \vspace{-1em}
\thanks{Atsu Kokuvi Angélo Passah is with ETIS Laboratory, ENSEA, CY Cergy Paris University, CNRS, France, and with the Pontifical Catholic University of Rio de Janeiro (PUC-Rio), Brazil.}%
\thanks{Rodrigo C. de Lamare is with the Pontifical Catholic University of Rio de Janeiro (PUC-Rio), Brazil, and with the School of Physics, Engineering and Technology, York University, United Kingdom.}%
\thanks{Arsenia Chorti is with ETIS Laboratory, ENSEA, CY Cergy Paris University, CNRS, France, and with Barkhausen Institut gGmbH, Germany.}
}


\maketitle

{
\begin{abstract}
Channel charting enables location-based services (LBSs) without requiring explicit position information by using pseudo-locations from the channel chart. While this property implies inherent privacy advantages, it does not provide formal privacy guarantees. 
In this work, we address location privacy in channel charting referred to as chart location indistinguishability (CLI), which extends geo-indistinguishability (GI) to channel charting representations. In order to achieve CLI, a standard planar Laplace mechanism is investigated and a geometry-aware Mahalanobis norm planar Laplace (MNPL) mechanism is devised. The proposed MNPL mechanism perturbs the channel chart by injecting noise aligned with the local structure of the chart. In the CLI framework with MNPL, privacy is defined in latent channel chart manifolds using locally adaptive covariance derived from chart neighborhoods, while preserving manifold topology under privacy constraints. In addition, differential privacy is considered as a privacy baseline. The proposed approach is evaluated across multiple channel charting schemes. The performance is assessed using utility metrics such as quality loss (QL) and range query error (RQE), as well as geometry-aware metrics including trustworthiness (TW) and continuity (CT). Numerical results demonstrate that the proposed privacy mechanism provides strong privacy guarantees while preserving the channel chart for LBSs tasks.  
\end{abstract}

}

\section{Introduction}
\IEEEPARstart{T}{he} increasing availability of channel state information (CSI) at base stations, in particular large-scale components, has enabled the development of channel charting (CC). CC is an unsupervised machine learning technique that maps high-dimensional channel features into a low-dimensional latent space while preserving the local geometry of the wireless environment \cite{8444621}. CC exploits the intrinsic structure of wireless channels to recover the relative spatial relationships between user equipments (UEs) without requiring ground-truth location labels. This makes CC particularly attractive for applications such as mobility management, handover optimization, and location-based services (LBSs) in general \cite{10155724, 10476798, 9482504, 9306087}. 

While CC enables location-based operations, sharing the learned chart embeddings with a third-party LBS introduces privacy risks. Although CC provides pseudo-locations and can be seen to provide a certain privacy level, it encodes physical locations of UEs. A curious LBS that has access to even a few anchor points can exploit the spatial consistency of the chart to build a mapping from the chart space to the physical space \cite{10070385}. This issue is related to the problem of location privacy. Typically, privacy is systematized through the framework of differential privacy (DP) \cite{Dwork11787006_1}, which in CC can be achieved by adding calibrated perturbations to the chart coordinates. Furthermore, focusing precisely on location privacy, the concept of geo-indistinguishability (GI) \cite{chatzikokolakis2015geo, Andr_s_2013}, extends the principles of DP to metric spaces such as the Euclidean space. GI provides a distance-dependent privacy guarantee by ensuring that nearby locations produce statistically similar outputs.

\subsection{Previous and Related Work}

The seminal work in \cite{8444621} proposed a CC framework where large-scale fading channel features were extracted from multi-antenna signals and then fed into the input of the charting algorithm. A more efficient alternative has been introduced in \cite{9526758} where phase-insensitive distance measures were introduced to reduce computational complexity while maintaining robustness against small-scale fading. Other advancements include the use of super-resolution angle-delay-power-profiles combined with the earth-mover distance (EMD), which outperform CSI features and linear algebraic distance measures \cite{10052099}.

A major challenge in CC is ensuring that the latent representation preserves the physical proximity relationships of UEs. Recent methods using timestamps for triplet-based learning have been proposed to improve performance \cite{9448128}. To further align the charts with the spatial geometry, split triplet losses and inertial regularizers have been proposed to account for the finite acceleration and velocity of UEs \cite{9562215}. Additionally, representation-constrained autoencoders can incorporate anchor vectors as side information to unwrap the channel chart and recover the global geometry of the propagation environment \cite{huang2019}.

While CC primarily preserves local neighborhood relationships, recent research focuses on achieving global geometric consistency and absolute localization. Multipoint CC enables distributed base stations to collaboratively learn a unified multi-cell radio map. This method exploits the redundancy in multipoint CSI to mitigate the distortion in single-point CC, thereby improving the accuracy for cell-edge UEs \cite{8645281}. The adoption of novel geodesic distances in \cite{10070385} allowed CC to approach a linear correlation with physical distance. This approach enables a global similarity in the chart. 

CC has also been proposed for online scenarios with high rates such as streaming CSI data. New continual learning strategies with dual-memory short-term buffers allow models to update dynamically and improve performance compared to existing replay-based approaches when dealing with non-stationary CSI \cite{11161077}. Moreover, the framework has been extended beyond static mapping to predictive modeling. By using joint-embedding predictive architectures (JEPA) and a suitable conditioning variable (UE velocity) in \cite{10777043}, the method can predict future latent states that capture the dynamics of the wireless system. These advancements have enabled novel applications such as unsupervised trajectory anomaly detection where unusual movement patterns are identified directly in the latent representation space \cite{11201239}. 

Concerning the GI privacy framework, the formal guarantee has been established in \cite{chatzikokolakis2015geo, Andr_s_2013} as an extension of DP tailored for location-based systems in a 2D Euclidean space. Research has expanded GI from 2D Euclidean planes to more complex spatial environments. A 3D-GI was introduced utilizing a three-variate Laplacian mechanism to simultaneously perturb the 3D coordinates \cite{9646489}. 

In the original GI framework, the same perturbation is applied to all locations without taking into account the specific privacy requirements that different locations may have. Recent advancements have addressed the limitations of uniform privacy budgets which may overprotect low sensitivity areas while underprotecting sensitive ones \cite{10815979, 10462511}. Location-discriminative geo-indistinguishability (LGDI) allows for varying privacy budgets based on location's sensitivity level such as distinguishing between a home address and a shopping mall \cite{10815979}. Similarly, semantic adaptive geo-indistinguishability
mechanism (SAGEO) utilizes semantic tags to adaptively adjust location indistinguishability without the need for a third party \cite{10462511}. The privacy sensitivity was modeled as a curved distance inspired by general relativity. 

A major vulnerability in GI-based systems is the linear increase of the privacy cost over multiple queries, known as the composition problem \cite{8125751, 10376293}. For frequent queries, improved GI-based location perturbation mechanisms use cell layouts and predictive functions to reduce the per-query privacy cost toward zero when the user is in both stationary and moving scenarios \cite{8125751}. In specific applications like location-based advertising (LBA) systems, \cite{10376293} proposed the PrivLocAd framework that mitigates long-term observation and multi-platform collusion attacks using longitudinal location obfuscation. 

GI has been integrated into diverse location-based services and sensing tasks. In mobile crowdsensing (MCS), a Q-learning-based approach was proposed to manage the continuously changing and uncertain task and worker states under GI \cite{11083355}. This approach optimizes worker travel distance while maximizing profit for both workers and the server. In social networking applications, GI has been applied to event retrieval by generating dummy locations based on real population data. This ensures queries are made in realistic inhabited areas rather than uninhabited regions \cite{10444179}.

\subsection{Contributions}
{
In this work, we investigate location privacy in CC, referred to as chart location indistinguishability (CLI) and extend geo-indistinguishability (GI) to CC representations. Both the standard planar Laplace (PL) mechanism used in GI and a geometry-aware Mahalanobis norm planar Laplace (MNPL) mechanism are investigated. MNPL incorporates the intrinsic geometry of the CC through the Mahalanobis distance \cite{2018ReprintOM}. 
The main contributions of this paper are summarized as follows.
\begin{itemize}
\item We propose a privacy framework for CC, termed CLI, which extends the concept of GI to CC representations. CLI defines privacy in latent CC manifolds using a locally adaptive covariance, derived from chart neighborhoods, while preserving the manifold topology under privacy constraints.

\item We investigate both the PL and the proposed MNPL mechanisms within the CLI framework. While PL provides an isotropic perturbation mechanism, MNPL incorporates a locally adaptive covariance to perturb the CC by injecting noise aligned with the local geometric structure of the chart. We then demonstrate that the proposed MNPL mechanism satisfies CLI.

\item 
Standard DP based on the Gaussian mechanism is also investigated as a privacy baseline.

\item Privacy and manifold topology preservation are evaluated across multiple CC schemes using utility metrics such as quality loss (QL) and range query error (RQE), as well as geometry-aware metrics including trustworthiness (TW) and continuity (CT). Numerical results demonstrate that the proposed privacy mechanism provides strong privacy guarantees while preserving the CC performance.
\end{itemize}
}

{
The remainder of the paper is organized as follows. The background on CC, DP and geo-distinguishability, is presented in Section \ref{background} followed by Section \ref{sys_mod} where the system model is described. The proposed CLI framework is included in Section \ref{cli_privacy} followed by the performance evaluation metrics in Section \ref{perf_metrics}. The proposed scheme is assessed and numerical results are presented in Section \ref{num_results} and the work is concluded in Section \ref{conclude}.  
}

\section{Background} \label{background}

This section presents the background on CC and geo-indistinguishability required for this work.

\subsection{Channel Charting} \label{cc_}
We first present the dimensionality reduction objective of CC, followed by two CC techniques based on Siamese networks \cite{8919897} and triplet-based learning \cite{9448128}.
\subsubsection{Local geometry preservation} CC is a framework that learns a mapping from channel features to a low-dimensional chart space in an unsupervised manner. The objective is to preserve the physical neighborhood structure between UEs in the obtained chart. Formally, given a set of 
$U$ channel feature vectors $\{\mathbf{f}_i\}_{i=1}^U$ extracted from CSI measurements, CC learns an encoder $\phi_\theta$ 
that maps each feature vector to low-dimensional embedding (i.e., $D = 2$) described by
\begin{equation}
    \mathbf{z}_i = \phi_\theta(\mathbf{f}_i) \in \mathbb{R}^D, 
    \label{eq_z}
\end{equation}
with $i = 1, \ldots, U$. 
The set of all embeddings $\{\mathbf{z}_i\}_{i=1}^U$ forms the CC.
The preservation of physical  distances in the CC is expressed as
\begin{equation}
    \|\mathbf{x}_i - \mathbf{x}_j\|_2 \approx  \,\|\mathbf{z}_i - \mathbf{z}_j\|_2
    \label{eq_cc}
\end{equation}
where $\mathbf{x}_i \in \mathbb{R}^2$ denotes the physical location of UE $i$.

Since ground truth positions are not available during training, CC relies on the spatial consistency of wireless channel features. This requires a good CSI transformation technique \cite{8444621} in order to accurately represent the true locations in the feature space. The proximity preservation between the feature space and the CC is then expressed as
\begin{equation}
    \mathcal{D}_f(\mathbf{f}_i,  \mathbf{f}_j) \approx  \,\|\mathbf{z}_i - \mathbf{z}_j\|_2,
\end{equation}
where $\mathcal{D}_f$ denotes the feature dissimilarity measure.

\subsubsection{Siamese network channel charting}
Channel features are mapped to a low-dimensional embedding space through a parametric Siamese network encoder $\phi_\theta(\cdot)$ (Fig. \ref{fig_siamese})  with 
$
\mathbf{z}_i = \phi_\theta(\mathbf{f}_i) \in \mathbb{R}^D
$.
\begin{figure}[!t]
    \centering
    \includegraphics[width=0.8\linewidth]{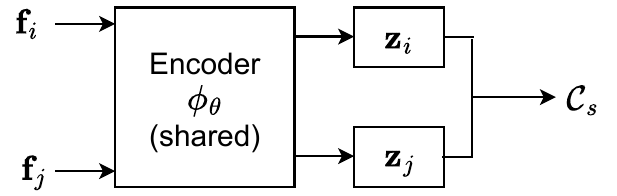}
    \caption{Siamese network using a sampling distance proxy to estimate the physical distance between the UE positions. Both features $\mathbf{f}_i$ and $\mathbf{f}_j$ are encoded with the same encoder.}
    \label{fig_siamese}
\end{figure}
The Siamese network consists of two identical encoder branches with shared parameters, where $\phi_\theta(\cdot)$ is a fully connected feedforward neural network that maps input features to a low-dimensional chart space.
Given a pair of channel feature vectors $(\mathbf{f}_i, \mathbf{f}_j)$, each branch processes one input and produces the corresponding low-dimensional embedding $\mathbf{z}_i$ and $\mathbf{z}_j$.
To enforce consistency between distances in the embedding space and estimated physical distances, we use a proxy for the physical distance between samples. Assuming that the UE moves smoothly and is sampled at a constant sampling rate, the physical distance between two samples can be approximated using their temporal indices as
\begin{equation}
d_s = |i - j| \Delta_d,
\end{equation}
where $\Delta_d$ denotes the average displacement per point.
The network is trained to align the Euclidean distance in the embedding space with this estimated physical distance. Specifically, for a set of sampled pairs $\mathcal{T}_s = \{(i, \, j)\}$, the training goal is to minimize the mean squared error between embedding distances and sampling-based distances. The cost function of this optimization problem is given by 
\begin{equation}
\mathcal{C}_s = \frac{1}{|\mathcal{T}_s|} \sum_{(i,j)\in\mathcal{T}_s} \left( \|\mathbf{z}_i - \mathbf{z}_j\|_2 - d_s \right)^2.
\label{eq_cost}
\end{equation}
This objective encourages the learned embedding to reflect the spatial structure of the environment by preserving relative distances between samples. 

\subsubsection{Triplet-based channel charting}
In triplet-based CC, instead of relying on predefined distances in the original feature space, which may be unreliable in high-dimensional settings, the approach learns a representation in which similar samples are mapped close together while dissimilar samples are mapped farther apart. In CC, similarity can be inferred from temporal proximity, enabling a self-supervised formulation \cite{9448128}.

The training data consist of triplets of CSI feature vectors of the form $(\mathbf{f}_i, \mathbf{f}_j, \mathbf{f}_k)$, where $\mathbf{f}_i$ denotes an anchor sample, $\mathbf{f}_j$ is a positive sample considered similar to the anchor and $\mathbf{f}_k$ is a negative sample considered dissimilar to $\mathbf{f}_i$. $\phi_\theta(\cdot)$ maps each input sample into a low-dimensional embedding space. In particular, each CSI feature vector is transformed into an embedding $\mathbf{z}_i = \phi_\theta(\mathbf{f}_i)$, $\mathbf{z}_j = \phi_\theta(\mathbf{f}_j)$ and $\mathbf{z}_k = \phi_\theta(\mathbf{f}_k)$, where all triplets share the same encoder $\phi_\theta(\cdot)$.

The learning objective is defined through the triplet loss:
\begin{equation}
\mathcal{C}_t = \dfrac{1}{|\mathcal{T}_t|}\sum_{(i,j,k) \in \mathcal{T}_t} \max \left( 0,\ \|\mathbf{z}_i - \mathbf{z}_j\|^2 - \|\mathbf{z}_i - \mathbf{z}_k\|^2 + \alpha \right),
\end{equation}
where $\alpha > 0$ is a margin parameter and $\mathcal{T}_t$ denotes the set of triplets \cite{9448128}.
This loss function encourages the embedding of the anchor sample to be closer to the positive sample than the negative sample by at least the margin $\alpha$, thereby preserving relative similarity relationships.

\subsection{Differential Privacy and Geo-indistinguishability}
{
Before dwelling into how GI extends the DP principle to the metric space, let  us first define DP.

\subsubsection{Differential Privacy} 
DP ensures that the outcome of a data analysis does not meaningfully depend on any single user's data by adding just enough randomness to protect individuals while keeping results useful \cite{Dwork11787006_1}. The formal definition of DP is given below.
\begin{definition}\label{def_dp_}
A randomized mechanism $\mathcal{M}$ is $(\varepsilon, \delta)$-differentially private if for any two datasets $\mathcal{D}_1$ and $\mathcal{D}_2$ that differ in at most one element and any output set $\mathcal{S} \subseteq Range(\mathcal{M})$, 
\begin{equation}
    Pr[\mathcal{M}(\mathcal{D}_1) \in \mathcal{S}] \leq e^{\varepsilon}\,Pr[\mathcal{M}(\mathcal{D}_2) \in \mathcal{S}] + \delta,
\end{equation}
where $\varepsilon$ is the privacy budget and $\delta$ is the failure probability. $\varepsilon$ is expected to be small (large value of $\varepsilon$ imply weaker privacy guarantee).
\end{definition}
The objective is to distort the data according to the privacy budget. $(\varepsilon, \delta)$-differential privacy is achieved by adding a calibrated noise to the output of the query function. A query function is a deterministic function that maps a dataset to an output without randomness. 
A standard way to achieve $(\varepsilon, \delta)$-DP is through the Gaussian mechanism, which perturbs the output of the query function by adding Gaussian noise. The variance $\sigma^2_{dp}$ of the Gaussian noise is given by 
\begin{equation}
    \sigma^2_{dp}
= \frac{2 \Delta^2 \ln(1.25/\delta)}{\varepsilon^2},
\end{equation}
where $\Delta$ is the sensitivity, which represents the maximum possible change over any two datasets that differ by exactly one record (neighboring datasets) \cite{10152847}.
}

\subsubsection{Geo-indistinguishability}
GI is a formal distance-dependent privacy framework designed for location-based systems, where the required privacy level varies with the physical distance between locations. The definition of GI is given in Definition \ref{def_geo}.

\begin{definition} \label{def_geo}
A randomized mechanism 
$\mathcal{M}$ satisfies 
$\varepsilon$-geo-indistinguishability if for all locations $\mathbf{z}_i$, $\mathbf{z}_j$ and all measurable sets $S$,

\begin{equation}
    \Pr[\mathcal{M}(\mathbf{z}_i) \in S] 
    \leq e^{\varepsilon \, d(\mathbf{z}_i, \mathbf{z}_j)} 
    \Pr[\mathcal{M}(\mathbf{z}_j) \in S]
    \label{eq:geo_indist}
\end{equation}
where $\varepsilon > 0$ is the privacy budget and $d(\mathbf{z}_i, \mathbf{z}_j)$ is the distance between the two locations.
\end{definition}

The core idea is that a user located at position $\mathbf{z}_i$ should enjoy $\ell$-privacy within a radius $r$, meaning that by observing the mechanism output 
$\mathcal{M}(\mathbf{z}_i)$, an adversary's ability to distinguish $\mathbf{z}_i$ from any other location $\mathbf{z}_j$ within the radius $r$ should not increase by more than a factor of $e^\ell$ with $\ell = \varepsilon\,r$. This definition ensures that the privacy guarantee depends on the distance, that is, locations that are close to each other require stronger privacy protection. $r$ is the desired protection radius and $\ell$ is the maximum acceptable privacy 
loss within that radius. The planar Laplacian (PL) mechanism was proposed in \cite{Andr_s_2013} to achieve GI.

Next, the system model is presented in Section \ref{sys_mod}, where the general scheme combining CC, the privacy framework and the released private data to the LBS server is described.

\section{System Model} \label{sys_mod}
Let us consider the system model shown in Fig. \ref{model} which represents the process in which CC is performed at the base station (BS) and the resulting pseudo-locations on the CC are perturbed and sent to the location-based service (LBS) server. The LBS server receives the perturbed chart coordinates and performs proximity, clustering, tracking, etc. depending on the context of the location-based application.
The system model includes the following stages:
    \subsubsection{Real locations} A set of users (UEs) with spatial locations reported by the coordinates  \((\mathbf{x}_1,\,\mathbf{x}_2, \, \ldots, \, \mathbf{x}_U) \), where $U$ denotes the number of users.
    
    \subsubsection{ Base station} The BS uses the received signal to estimate the CSI of all users across the receive antennas. Typically, the BS can derive from the CSI the received signal strength (RSS), the angles of arrival, and the propagation delays. The BS then performs feature extraction to retrieve useful large-scale properties from the CSI for CC. Random CSI fluctuations (e.g., due to small-scale fading) may distort neighborhood relationships making feature extraction a critical step for CC. The extracted CSI features \((\mathbf{f}_1,\,\mathbf{f}_2, \, \ldots, \, \mathbf{f}_U) \) should mimic the local spatial relationships between UEs.

    \subsubsection{Feature engineering and charting function}
    CC consists of two main stages: feature engineering from CSI measurements and the learning of a charting function that maps these features into a low-dimensional (i.e., $D = 2$) embedding space.
The overall processing chain is illustrated in Fig. \ref{feat_ext}. Starting from the CSI estimated at the BS, a sequence of transformations is applied to obtain robust feature vectors that accurately represent the spatial locations. 
\begin{figure}[!t]
    \centering
    \includegraphics[width=0.8\linewidth]{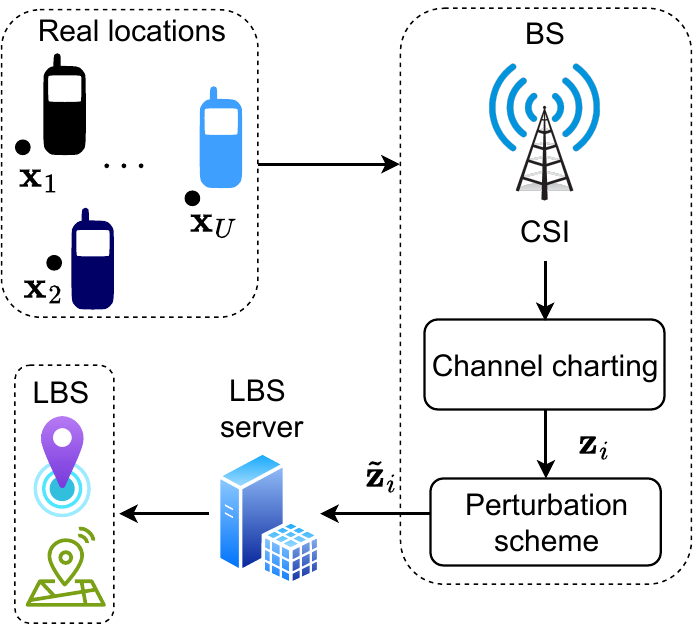}
    \caption{Block diagram of the system model. The base station performs channel charting after estimating the CSI from the real location coordinates. Then the obtained chart coordinates are perturbed and released to the LBS server.}
    \label{model}
\end{figure}

\begin{figure}[!t]
    \centering
    \includegraphics[width=0.95\linewidth]{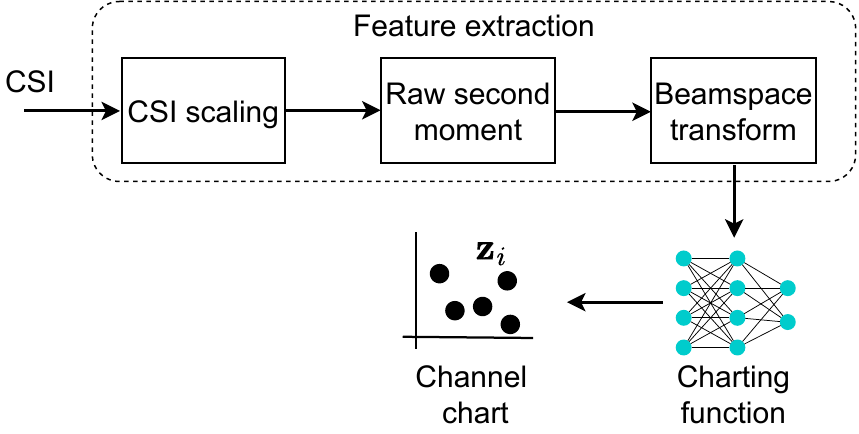}
    \caption{Channel Charting. The CSI estimation is followed by a feature extraction step including the scaling of the CSI, the computation of the second moment and the beamspace transformation. The obtained features are fed into the input of the encoder function that outputs the low-dimensional channel chart.}
    \label{feat_ext}
\end{figure}
Let $\mathbf{h}_i \in \mathbb{C}^{N}$ denote the CSI vector corresponding to the $i$-th UE, where $N$ is the number of antennas. The feature engineering process consists of three steps: feature scaling, second-order moment computation, and beamspace transformation, similarly to the work in \cite{8444621}. In detail, the following steps are typical in CC: 
\begin{itemize} 
\item \textit{CSI scaling}: 
Since CSI in radio geometry is a poor representation of spatial geometry, CSI scaling allows to correct it as shown in \cite{8444621}. The CSI vector is typically scaled as 
\begin{equation}
\bar{\mathbf{h}}_i =
\frac{N^{\beta - 1}}{\|\mathbf{h}_i\|_2^{\beta}} \mathbf{h}_i,
\label{eq_csi_scale}
\end{equation}
where $\beta = 1 + \dfrac{1}{2\gamma}$ and  $\gamma > 0$ is a design parameter related to the path loss.

\item \textit{Raw second moment evaluation}:
Next, a second order statistical moment representation of the scaled CSI is constructed as
\begin{equation}
\mathbf{R}_i = \bar{\mathbf{h}}_i \bar{\mathbf{h}}_i^H \in \mathbb{C}^{N \times N},
\end{equation}
where $(\cdot)^H$ denotes the Hermitian transpose.
The matrix $\mathbf{R}_i$ captures the spatial correlations between antenna elements and encodes the angular and propagation characteristics of the wireless channel.

\item \textit{Beamspace transformation}:
To enhance the spatial structure of the features representation, the correlation matrix is transformed into the beamspace domain using the discrete Fourier transform (DFT). Let $\mathbf{T} \in \mathbb{C}^{N \times N}$ denotes the unitary DFT matrix such that $\mathbf{T}\mathbf{T}^H = \mathbf{I}$. The beamspace representation is given by $\mathbf{T} \mathbf{R}_i \mathbf{T}^H$, and 
\begin{equation}
\mathbf{K}_i = \left| \mathbf{T} \mathbf{R}_i \mathbf{T}^H \right|,
\end{equation}
where $|\cdot|$ denotes the element-wise magnitude.
The matrix $\mathbf{K}_i$ is subsequently vectorized to obtain the feature vector
\begin{equation}
\mathbf{f}_i = \operatorname{vec}(\mathbf{K}_i) \in \mathbb{R}^{N^2}.
\end{equation}
\end{itemize}

The obtained CSI feature vectors $\{\mathbf{f}_i\}_{i=1}^U$ are then fed into the input of the charting function model that learns local neighborhood relationships in an unsupervised manner. The goal of this stage is to output the CC coordinates \((\mathbf{z}_1,\,\mathbf{z}_2, \, \ldots, \, \mathbf{z}_U)\) while preserving the local geometry of neighboring locations in the spatial domain as stated in \ref{cc_}, Eq. \ref{eq_cc}. The obtained chart locations are perturbed and released to the untrusted LBS server. 
    \subsubsection{LBS server and LBS} The untrusted LBS server performs the location-based operations using the perturbed channel chart.

Perturbing the obtained chart coordinates can substantially reduce the quality of service (QoS) at the LBS. Therefore, the objective is to design the perturbation scheme in such a way that a good QoS is obtained while preserving the privacy of the chart locations.
We propose a privacy mechanism that satisfies the properties of geo-indistinguishability to strengthen the privacy protection of the chart coordinates. In the following, we detail the proposed chart location indistinguishability framework. 

\section{Chart Location Indistinguishability} \label{cli_privacy}

In the context of CC, the notion of physical location is replaced by chart coordinates that capture the relative geometry of UE locations. To perturb the CC before releasing it to the LBS server, GI that uses the planar Laplace mechanism can be applied in the chart space. While GI provides a privacy guarantee using Euclidean distance, it does not directly account for the intrinsic structure of CCs, in particular the neighborhood structure of the CC. In this section, we first propose the use of standard GI in the chart space using the planar Laplacian. Then, the proposed chart location indistinguishability based on the Mahalanobis distance \cite{2018ReprintOM} is presented. CLI extends the concept of GI by enforcing a distance-based indistinguishability guarantee between chart points, while allowing the distance metric to reflect the underlying geometry of the chart.

\begin{definition}\label{def_CLI} A randomized mechanism $\mathcal{M}$ satisfies $\varepsilon$-chart location indistinguishability ($\varepsilon$-CLI), if for all chart locations $\mathbf{z}, \mathbf{z}^{\prime} \in \mathbb{R}^2$ and all measurable sets $\mathcal{S}$,
\begin{equation}
\Pr[\mathcal{M}(\mathbf{z}) \in \mathcal{S}]
\leq
e^{\varepsilon \, d_{c}(\mathbf{z}, \mathbf{z}^{\prime})}
\Pr[\mathcal{M}(\mathbf{z}^{\prime}) \in \mathcal{S}],
\end{equation}
where $\varepsilon > 0$ is the privacy budget and $d_{c}(\cdot,\cdot)$ denotes a distance metric defined in the chart space.
\end{definition}

This definition ensures that the probability distributions of the perturbed two nearby chart points are similar, making them difficult to distinguish. The level of indistinguishability is controlled by their distance in the chart. The metric $d_{c}$ can incorporate the geometry of the chart, enabling more flexible and structured privacy guarantees.

\subsection{Planar Laplace Mechanism}
To achieve $\varepsilon$-CLI, one can directly apply the standard planar Laplace mechanism in the chart space. In this case, the distance $d_{c}$ is chosen as the Euclidean distance between chart coordinates. This approach provides a direct extension of geo-indistinguishability to channel charts without accounting for their local geometric structure. In this context, the chart locations are perturbed uniformly in all directions and the perturbation is independent of directions. 

\subsubsection{Planar Laplacian}
The planar Laplacian \cite{Andr_s_2013} is a generalization of the Laplace mechanism \cite{Dwork} to two dimensions. It is the obfuscation mechanism used to enforce geo-indistinguishability in the Euclidean plane. Given a location $\mathbf{z} \in \mathbb{R}^2$, the mechanism releases a perturbed coordinate $\tilde{\mathbf{z}}$ with a two-dimensional Laplace distribution noise. Given the actual chart location $\mathbf{z}$, the probability density function is given by 
\begin{equation}
f_{\varepsilon}(\mathbf{z})(\tilde{\mathbf{z}}) 
=\frac{\varepsilon^2}{2\pi}e^{-\varepsilon\|\tilde{\mathbf{z}}-\mathbf{z}\|_2}, 
\end{equation}
where $\frac{\varepsilon^2}{2\pi}$ is a normalization factor. 

\subsubsection{Drawing a random point of PL noise} For convenience, this probability distribution is then transformed into a system of polar coordinates $(r, \rho)$. In this representation, the angle is uniformly distributed over $[0,\,2\pi)$ while the radius is sampled from a Gamma distribution \cite{Andr_s_2013}. By transforming the density in polar coordinates, the angle and the radius can be drawn independently and then construct the noise components. Let the chart location coordinate be given by $\mathbf{z} = (z_x, z_y)$, then the released perturbed coordinate is given by $\tilde{\mathbf{z}} = (\tilde{z}_x, \tilde{z}_y)$ such that $\tilde{z}_x = z_x + r\cos\rho$ and $\tilde{z}_y = z_y + r\sin\rho$.


\subsection{Mahalanobis Norm Planar Laplace Mechanism}
While PL ensures a valid $\varepsilon$-CLI guarantee in the sense of GI, it may lead to suboptimal QoS performance in CC applications, where the geometry is locally structured. To better exploit the neighborhood structure of CCs, we propose a geometry-aware extension of the planar Laplace mechanism based on the Mahalanobis distance \cite{2018ReprintOM}. In this approach, the distance $d_{c}$ is defined through a positive definite covariance matrix that captures local geometric properties of the chart. This enables different perturbation levels along different directions.

This mechanism provides an $\varepsilon$-CLI guarantee with respect to the Mahalanobis distance by generalizing PL to a geometry-aware setting. By adapting the perturbation to the local geometry of the chart, the Mahalanobis norm planar Laplace (MNPL) mechanism achieves a better trade-off between privacy and utility. Hence, neighborhood relationships are preserved while ensuring strong indistinguishability guarantees.

\subsubsection{Mahalanobis norm} The Mahalanobis distance is a generalized distance metric that accounts for the correlation structure of the dataset. Given a symmetric positive definite (PD) matrix $\mathbf{\Sigma}$, the Mahalanobis distance between two points $\mathbf{z}$ and $\mathbf{z}^{\prime}$ is given by 
\begin{equation}
d_M(\mathbf{z}, \mathbf{z}^{\prime}) = \sqrt{(\mathbf{z}-\mathbf{z}^{\prime})^\top \mathbf{\Sigma}^{-1}(\mathbf{z}-\mathbf{z}^{\prime})}.
\end{equation}
This can also be written as \(d_M(\mathbf{z}, \mathbf{z}^{\prime})=\|\mathbf{\Sigma}^{-1/2}(\mathbf{z}-\mathbf{z}^{\prime})\|_2\), where $\mathbf{\Sigma}$ can be interpreted as a PD matrix that defines the geometry of the space.

\subsubsection{MNPL probability density function (PDF)}
Given the actual chart location $\mathbf{z}$, the probability density function of the MNPL mechanism is given by 
\begin{equation}
p_{\varepsilon}(\mathbf{z})(\tilde{\mathbf{z}}) 
 = \frac{\varepsilon^2}{2\pi\sqrt{\mathrm{det}(\mathbf{\Sigma})}}
e^{-\varepsilon \, d_M(\tilde{\mathbf{z}}, \, \mathbf{z})},
\label{eq_mnpl}
\end{equation}
where $d_M(\tilde{\mathbf{z}}, \mathbf{z}) = \|\mathbf{\Sigma}^{-1/2}(\tilde{\mathbf{z}}-\mathbf{z})\|_2$ corresponds to the distance metric $d_{c}(\cdot,\cdot)$ mentioned in Definition \ref{def_CLI} of $\varepsilon$-CLI. 
Using the change-of-variable rule based on the Jacobian determinant, it can be shown that the integral of $p_{\varepsilon}(\mathbf{z})(\tilde{\mathbf{z}})$ over $\mathbb{R}^2$ is equal to $1$. Thus, $p_{\varepsilon}(\mathbf{z})(\tilde{\mathbf{z}})$ is a probability density function.
This MNPL PDF reduces to the PL PDF when $\mathbf{\Sigma}$ is the identity matrix ($\mathbf{\Sigma} = \mathbf{I}$).  

\subsubsection{Construction of the covariance matrix}
$\mathbf{\Sigma}$ is constructed from the local geometry of the chart locations to ensure that the perturbation is aligned with the geometry of the chart. 
Let $\mathcal{N}_K(\mathbf{f})$ be the set of the $K$ nearest neighbors of the channel feature $\mathbf{f}$. The corresponding chart coordinates of these neighbors are used to compute a local covariance matrix in the chart space. Then, the eigenvalue decomposition of this covariance matrix yields two orthogonal directions associated with the first two largest eigenvalues. Let $\mathbf{u}_p$ and $\mathbf{u}_o$ be orthonormal direction vectors derived from these orthogonal directions, such that $\mathbf{u}_p$ is the principal direction and  $\mathbf{u}_o$ is the orthogonal one (second principal direction). From $\mathbf{u}_p$ and $\mathbf{u}_o$, $\mathbf{\Sigma}$ is obtained as
\begin{equation}
\mathbf{\Sigma} = \sigma_p^2 \mathbf{u}_p \mathbf{u}_p^\top + \sigma_o^2 \mathbf{u}_o \mathbf{u}_o^\top,
\end{equation}
where $\sigma_p^2$ and $\sigma_o^2$ are design parameters which control the amount of perturbation along each direction. 
To ensure that $\mathbf{\Sigma}$ is a positive definite matrix, $\sigma_p^2$ and $\sigma_o^2$ should be strictly positive ($\sigma_p^2$ and $\sigma_o^2 > 0$). 
$\mathbf{\Sigma}$ should be interpreted as a geometry-aware perturbation matrix rather than as the empirical covariance itself. 
This construction enables the mechanism to exploit the local geometry of the chart while controlling the privacy-induced distortion.

Therefore, the perturbation depends on the parameters $\varepsilon$, $\sigma_p^2$ and $\sigma_o^2$. The privacy budget $\varepsilon$ controls the global magnitude of the perturbation noise while $\sigma_p^2$ and $\sigma_o^2$ control how the perturbation is distributed across directions. 
Since the local distribution exhibits higher variance along the principal direction and lower dispersion along the orthogonal direction, the chart locations are more tightly distributed along the orthogonal direction.
Therefore, stronger perturbation should be applied along the direction where chart locations are less dispersed by choosing $\sigma_o > \sigma_p$. 
This design is consistent with the distance-dependent principle of GI where stronger protection is required for nearby locations. 

In practice, $\sigma_p^2$ and $\sigma_o^2$ should remain consistent with the scale of the CC in order to avoid perturbations that dominate the manifold geometry. Assuming the chart coordinates are normalized, it is natural to consider $\sigma_p, \sigma_o \in (0,1]$ in order to keep the perturbation magnitude comparable to the chart scale.
When $\sigma_p$ and $\sigma_o$ are sufficiently small relative to the normalized scale, the perturbation remains limited and the local manifold structure is largely preserved, resulting in limited utility degradation but weaker obfuscation of the chart locations.
In contrast, larger values may produce perturbations whose magnitude dominates the manifold geometry, causing significant distortion of neighborhood relationships and degradation of the chart geometry. 
The particular case $\sigma_p = \sigma_o = 1$ corresponds to $\mathbf{\Sigma} = \mathbf{I}$, where the MNPL mechanism reduces to the PL mechanism.


The choice of the neighborhood size $K$ affects the estimation of the local geometry during the construction of $\mathbf{\Sigma}$. A small value of $K$ allows the covariance matrix to better capture fine local structures of the chart manifold. However, very small neighborhoods may result in unstable covariance estimates. In contrast, larger values of $K$ provide smoother and more robust covariance estimates by averaging over a larger set of neighboring points, at the expense of a less accurate representation of the local geometry. $K$ should remain sufficiently small relative to the dataset size in order to preserve the local manifold structure while ensuring reliable estimation of the principal directions. A practical choice is approximately $5\%$ of the total number of points \cite{8444621, Venna_Jarkko}.

Although the proposed MNPL mechanism is described for $D = 2$ dimensional CCs in this work, the framework naturally generalizes to higher-dimensional latent representations. The local covariance matrix can generally be given by \(\mathbf{\Sigma} = \sum_{i = 1}^{D} \sigma_i^2\mathbf{u}_i\mathbf{u}_i^{\top}\), where $\mathbf{u}_i$ are orthonormal principal directions derived from the local covariance estimation and $\sigma_i^2$ are perturbation parameters controlling the perturbation magnitude along each direction.

\subsubsection{Privacy Guarantee of MNPL}
Theorem \ref{theo_cli} proves that the MNPL mechanism satisfies $\varepsilon$-chart location indistinguishability.
\begin{theorem} \label{theo_cli}
    The Mahalanobis norm planar Laplace mechanism satisfies $\varepsilon$-chart location indistinguishability ($\varepsilon$-CLI).
\end{theorem}
\begin{proof}
Let $\mathbf{z}$ and $\mathbf{z}^{\prime}$ be two points on the chart. From Eq. \eqref{eq_mnpl}, we have the probability density functions (PDF)
\begin{equation}
p_{\varepsilon}(\mathbf{z})(\tilde{\mathbf{z}})
 = \frac{\varepsilon^2}{2\pi\sqrt{\mathrm{det}(\mathbf{\Sigma})}}
e^{-\varepsilon \, d_M(\tilde{\mathbf{z}}, \, \mathbf{z})}
\end{equation}
and 
\begin{equation}
p_{\varepsilon}(\mathbf{z}^{\prime})(\tilde{\mathbf{z}})
 = \frac{\varepsilon^2}{2\pi\sqrt{\mathrm{det}(\mathbf{\Sigma})}}
e^{-\varepsilon \, d_M(\tilde{\mathbf{z}}, \,\mathbf{z}^{\prime})}.
\end{equation}
Hence, we have
\begin{equation}
\dfrac{p_{\varepsilon}(\mathbf{z})(\tilde{\mathbf{z}})}{p_{\varepsilon}(\mathbf{z}^{\prime})(\tilde{\mathbf{z}})} = e^{
\varepsilon\,\big(
d_M(\tilde{\mathbf{z}},\,\mathbf{z}^{\prime}) - d_M(\tilde{\mathbf{z}},\,\mathbf{z})
\big)}.
\end{equation}
Using the reverse triangle inequality, we have
\(
\big|
d_M(\tilde{\mathbf{z}},\mathbf{z}^{\prime}) - d_M(\tilde{\mathbf{z}}, \mathbf{z})\big|\le
d_M(\mathbf{z},\mathbf{z}^{\prime})\) $\implies$ \(
 d_M(\tilde{\mathbf{z}},\mathbf{z}^{\prime}) - d_M(\tilde{\mathbf{z}}, \mathbf{z}) \le d_M(\mathbf{z},\mathbf{z}^{\prime})\). By substituting this bound into the density ratio, we have 
 \begin{equation}
 \dfrac{p_{\varepsilon}(\mathbf{z})(\tilde{\mathbf{z}})}{p_{\varepsilon}(\mathbf{z}^{\prime})(\tilde{\mathbf{z}})} \le e^{\varepsilon \,d_M(\mathbf{z},\,\mathbf{z}^{\prime})},
 \end{equation}
 that is, 
 \begin{equation}
 p_{\varepsilon}(\mathbf{z})(\tilde{\mathbf{z}}) \le  e^{\varepsilon \,d_M(\mathbf{z},\,\mathbf{z}^{\prime})} p_{\varepsilon}(\mathbf{z}^{\prime})(\tilde{\mathbf{z}}) 
 \end{equation}
By integrating over any measurable set $\mathcal{S}$, we obtain
 \begin{equation}
 \int_{\mathcal{S}}p_{\varepsilon}(\mathbf{z})(\tilde{\mathbf{z}})ds \le  e^{\varepsilon \,d_M(\mathbf{z},\,\mathbf{z}^{\prime})} \int_{\mathcal{S}}p_{\varepsilon}(\mathbf{z}^{\prime})(\tilde{\mathbf{z}}) ds
 \end{equation}
and taking into account the definition of the mechanism probabilities, we have 
$\Pr[\mathcal{M}(\mathbf{z}) \in \mathcal{S})  = \int_{\mathcal{S}}p_{\varepsilon}(\mathbf{z})(\tilde{\mathbf{z}})ds$ and $\Pr[\mathcal{M}(\mathbf{z^{\prime}}) \in \mathcal{S}) = \int_{\mathcal{S}}p_{\varepsilon}(\mathbf{z}^{\prime})(\tilde{\mathbf{z}})ds$.
Therefore, we obtain
\begin{equation}
\Pr[\mathcal{M}(\mathbf{z}) \in \mathcal{S}]
\leq e^{\varepsilon \, d_{M}(\mathbf{z},\, \mathbf{z}^{\prime})}
\Pr[\mathcal{M}(\mathbf{z}^{\prime}) \in \mathcal{S}].
\end{equation}
\end{proof}

\subsubsection{Drawing a random point of MNPL noise} Sampling from the MNPL distribution can be efficiently performed using a transformation based on the standard PL noise. Specifically, a noise vector is first sampled from the PL distribution and a linear transformation is then applied as in Proposition \ref{propo_v}.
\begin{proposition} \label{propo_v}
    The noise vector $\mathbf{v}$ from the MNPL mechanism can be sampled as follows:
    \begin{itemize}
        \item Sample a noise $\mathbf{w}$ from the PL mechanism.
        \item Then set  $\mathbf{v} = \mathbf{\Sigma}^{1/2}\mathbf{w}$.
    \end{itemize}
\end{proposition}
\begin{proof}
First, let's set the perturbation as $\tilde{\mathbf{z}} = \mathbf{z} + \mathbf{v}$ where $\mathbf{v}$ is the noise vector from the MNPL distribution. The PDF of $\mathbf{v}$ can be written as 
\begin{equation}
p_V(\mathbf{v}) = \frac{\varepsilon^2}{2\pi\sqrt{\det(\mathbf{\Sigma})}}
e^{-\varepsilon \, d_M(\mathbf{v}, \mathbf{0})},
\label{eq_v}
\end{equation}
where $d_M(\mathbf{v}, \mathbf{0}) = \|\mathbf{\Sigma}^{-1/2} \mathbf{v}\|_2$.
By drawing the noise $\mathbf{w}$ from the PL distribution, we have 
\begin{equation}
p_W(\mathbf{w}) = \frac{\varepsilon^2}{2\pi}
e^{-\varepsilon \|\mathbf{w}\|_2}.
\end{equation}
Also, $\mathbf{v} = \mathbf{\Sigma}^{1/2}\mathbf{w}$ $\implies$ $\mathbf{w} = \mathbf{\Sigma}^{-1/2} \mathbf{v}$. Using the change of variables formula based on the determinant of the Jacobian, we have 
\begin{equation}
p_V(\mathbf{v}) = p_W(\mathbf{\Sigma}^{-1/2} \mathbf{v})|\det(\mathbf{\Sigma}^{-1/2})|. 
\end{equation}
Hence, we have
\begin{equation}
p_V(\mathbf{v}) =
 \dfrac{\varepsilon^2}{2\pi\sqrt{\det(\mathbf{\Sigma})}} e^{-\varepsilon \|\mathbf{\Sigma}^{-1/2} \mathbf{v}\|_2}
\end{equation}
We obtain $p_V(\mathbf{v}) = \frac{\varepsilon^2}{2\pi\sqrt{\det(\mathbf{\Sigma})}}
e^{-\varepsilon \, d_M(\mathbf{v}, \mathbf{0})}$ as in (\ref{eq_v}).
\end{proof}

{
\subsubsection{Computational complexity}
The dominant computational cost of the proposed MNPL mechanism arises from the $K$-nearest neighbor search performed in the feature space. For a CSI feature of dimension $N^2$, an exhaustive $K$-nearest neighbor search requires a computational complexity of $\mathcal{O}(UN^2)$ \cite{Ray_Susmita}, leading to a total complexity of $\mathcal{O}(U^2N^2)$ for all chart locations.
In contrast, the covariance estimation and eigenvalue decomposition steps are lightweight since they are performed in the low-dimensional chart space.
}

\section{Performance Metrics} \label{perf_metrics}

Evaluating privacy-preserving mechanisms requires a careful assessment of the trade-off between privacy protection and data utility. In the context of CC, this trade-off is particularly important, as the goal is not only to protect sensitive location information, but also to preserve the intrinsic geometric structure of the chart. To this end, we consider a set of complementary performance metrics that capture both classical utility aspects and geometry-aware properties of the perturbed chart locations.

In addition, we employ standard metrics commonly used in the GI literature such as quality loss (QL) \cite{Bordenabe_2014} and range query error \cite{Gu2019LDP}. These metrics quantify the distortion introduced by the perturbation mechanism as well as its impact on spatial query accuracy. 

First, we consider geometry-aware metrics, namely trustworthiness (TW) and continuity (CT), which are widely used in manifold learning and dimensionality reduction. These metrics evaluate how well local neighborhood relationships are preserved after perturbation, thereby capturing the extent to which the intrinsic geometry of the chart is maintained. 

\subsection{Trustworthiness and Continuity}
To evaluate how well the noisy chart preserves the local structure of the feature space, we employ TW and CT metrics.
For a fixed neighborhood size $K$, we define $\mathcal{N}_K(\mathbf{f}_i)$ as the set of $K$-nearest neighbors of $\mathbf{f}_i$ in the feature space and $\widetilde{\mathcal{N}}_K(\mathbf{f}_i)$ as the corresponding set in the perturbed chart. Let $k_i(u)$ denote the rank of point $u$ with respect to $\mathbf{f}_i$ in the feature space and $\widetilde{k}_i(u)$ its rank in the perturbed chart.

The trustworthiness metric is defined as:
\begin{equation}
\mathrm{TW} = 1 - \frac{1}{B U} \sum_{i=1}^{U} \sum_{u \in \widetilde{\mathcal{N}}_K(\mathbf{f}_i)} \left(k_i(u) - K\right)^{+}
\label{eq_tw}
\end{equation}
and the continuity metric is defined as:
\begin{equation}
\mathrm{CT} = 1 - \frac{1}{B U} \sum_{i=1}^{U} \sum_{u \in \mathcal{N}_K(\mathbf{f}_i)} \left(\widetilde{k}_i(u) - K\right)^{+},
\label{eq_ct}
\end{equation}
where $(x)^{+} = \max(x, 0)$ and the normalization constant $B$ is given by
\(
B = \dfrac{K(2U - 3K - 1)}{2}.
\)

Trustworthiness penalizes points that appear among the $K$ nearest neighbors in the perturbed chart but are not among the top $K$ neighbors in the feature space. Continuity penalizes points that are close in the feature space but are not preserved as neighboring points in the perturbed chart. Both metrics take values in $[0,1]$, where higher values indicate better preservation of the local neighborhood structure.

\subsection{Range Query Error}
In LBSs, a range query consists of retrieving all users or points of interest located within a given radius from a reference location. The range query error (RQE) evaluates the impact of the perturbation mechanism on spatial query accuracy. 
After perturbation, the reported locations may fall outside their true spatial neighborhoods leading to inaccurate query results. RQE quantifies this degradation by measuring how often perturbed locations fail to remain within a specified distance from their original positions.

Let $\mathbf{z}_i \in \mathbb{R}^2$ denote the true chart location of sample $i$ and $\tilde{\mathbf{z}}_i$ its perturbed version. For a given query radius $r > 0$, the range query centered at $\mathbf{z}_i$ corresponds to the set of points located within a distance $r$ from $\mathbf{z}_i$. RQE is defined as the proportion of perturbed locations that fall outside this query region,
\begin{equation}
\mathrm{RQE} = 1 - \frac{1}{U} \sum_{i=1}^{U} \mathbb{I} \left\{ \|\tilde{\mathbf{z}}_i - \mathbf{z}_i\|_2 \leq r
\right\},
\label{eq_rqe}
\end{equation}
where $\mathbb{I}\{ \cdot\}$ denotes the indicator function, which equals to $1$ if $\|\tilde{\mathbf{z}}_i - \mathbf{z}_i\|_2 \leq r$ and $0$ otherwise.
This metric takes values in $[0,1]$, where lower values indicate better utility. A small RQE means that most perturbed locations remain within the query radius of their true positions, thereby preserving the accuracy of spatial queries. Conversely, a higher value reflects a greater degradation of query results due to the applied perturbation.

\subsection{Quality Loss}
The quality loss (QL) \cite{Bordenabe_2014} measures the degradation of service quality induced by the perturbation mechanism. It quantifies the discrepancy between the true and perturbed chart locations. In more detail, let $\mathcal{Z}$ denote the set of possible chart locations and $\pi(\mathbf{z})$ a prior distribution over $\mathcal{Z}$, such that \(\sum_{\mathbf{z} \in \mathcal{Z}} \pi(\mathbf{z}) = 1\). Let $p(\tilde{\mathbf{z}} \mid \mathbf{z})$ denote the perturbation mechanism that represents the probability of reporting $\tilde{\mathbf{z}}$ when the true chart location is $\mathbf{z}$. Given a Euclidean distance-based quality metric $d(\cdot,\cdot)$ the quality loss is defined as
\begin{equation}
\mathrm{QL}
=
\sum_{\mathbf{z}, \,\tilde{\mathbf{z}}} \pi(\mathbf{z}) \, p(\tilde{\mathbf{z}} \mid \mathbf{z}) \, d(\mathbf{z}, \tilde{\mathbf{z}}).
\label{eq_ql}
\end{equation}
QL defines the expected distance between the real and reported locations.
A lower value of QL indicates better utility, as the perturbed locations remain closer to original chart locations.

\section{Numerical Results} \label{num_results}
{
In this section, we describe the dataset used to evaluate the proposed approaches, the simulation setup and carry out a comparison of the considered channel charting techniques and the impact of privacy mechanisms.
}

\subsection{Simulation Setup}
\subsubsection{Dataset} 
The measurement campaign \cite{Shehzad} was conducted in an urban environment characterized by buildings of approximately $\SI{15}{\meter}$ in height aligned along the streets. The antenna array was installed on the rooftop of one of the buildings. These structures served both as reflectors and blockers for the radio signals. Thus, the propagation conditions includes both line-of-sight (LoS) and non-line-of-sight (NLoS) scenarios.

The transmit array consisted of $64$ elements configured in four rows of $16$ single-polarized patch antennas. The elements were spaced by $\lambda/2$ in the horizontal direction and $\lambda$ vertically, where $\lambda$ is the wavelength. The transmission was performed at a carrier frequency of $\SI{2.18}{\giga\hertz}$ using OFDM waveforms. 

The pilot design enabled channel sounding over $50$ distinct subbands, each containing $12$ consecutive subcarriers within a duration of $\SI{0.5}{\milli\second}$. During this time interval, the channel was assumed to be time-invariant. These pilot bursts were transmitted periodically with a periodicity of $\SI{0.5}{\milli\second}$.

On the receiver side, a setup emulating user equipment (UE) was employed including a single monopole antenna positioned at a height of $\SI{1.5}{\meter}$. Frequency synchronization between the transmitter and receiver was achieved using GPS. During the measurements, the receiver was mounted on a movable cart and guided along multiple predefined routes at an approximate walking speed of $\SI{5}{\kilo\meter\per\hour}$. This resulted in a spatial sampling resolution less than $\SI{0.1}{\milli\meter}$.

\subsubsection{Encoder $\phi_\theta(\cdot)$}
The fully connected feedforward neural network encoder $\phi_\theta(\cdot)$ is given in Fig. \ref{neural_}. This consists of a sequence of dense layers, each followed by batch normalization (BN) and a ReLU activation function except the final layer. The final layer is a linear projection that maps the last hidden representation to the $D$-dimensional embedding space without any activation function, ensuring that the resulting embeddings are unconstrained in $\mathbb{R}^D$. 
\begin{figure}[!t]
    \centering
    \includegraphics[width=0.90\linewidth]{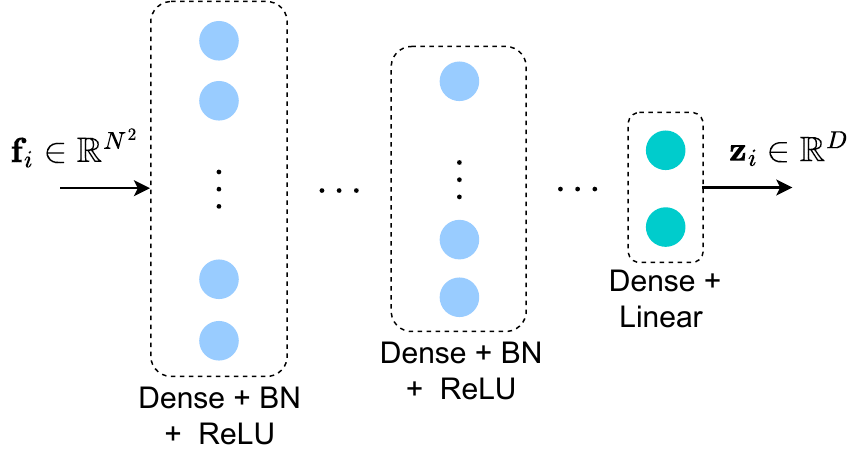}
    \caption{Encoder: fully connected feedforward neural network with each layer is includes a dense layer followed by a batch normalization (BN) and ReLU activation function, except the last embedding layer without BN and activation.}
    \label{neural_}
\end{figure}

The detailed structure of $\phi_\theta(\cdot)$ is given in Table \ref{table_encoder}. The input size $N^2$ is the length of the channel feature vector and the output dimension is the channel chart dimension $D = 2$. A ReLu activation is applied to all the layers except the last layer. 
For the learning, we use a batch size of $256$ and an Adam optimizer with a learning rate $ 10^{-3}$ over $100$ epochs.

\subsubsection{Simulation parameters}
We apply antenna subsampling with a factor of $8$  and the CSI data are averaged over the $50$ subbands. We compare four channel charting methods in the context of privacy schemes, including triplet-based learning (Triplet), autoencoder (AE), PCA and Siamese networks. To evaluate the proposed privacy mechanisms, we use the performance metrics presented in Section \ref{perf_metrics} including continuity given in \eqref{eq_ct} and trustworthiness given in \ref{eq_tw} to assess the charting performance, followed by quality loss given in \eqref{eq_ql} and range query error given in \eqref{eq_rqe}. These metrics are evaluated against different parameters such as the privacy budget $\varepsilon$, the radius $r$, the privacy loss $\ell$ and the number of nearest neighbors $K$.
Unless otherwise specified, the simulation parameters are defined as follows. The number of nearest neighbors is $K = 50$, the privacy budget $\varepsilon = 2$, the protection radius $r = 0.5$, the privacy loss $\ell = 0.5$ and the parameters $\sigma_p = 0.01$ and $\sigma_o = 0.05$.  

\begin{table}[!t]
\centering
\caption{Network structure of the encoder in Fig. \ref{neural_}.}
\renewcommand{\arraystretch}{1.5}
\begin{tabular}{|c|c|c|}
\hline
Layer & Dimensions & Activation \\
\hline
Input & $N^2$  & Linear \\
\hline
$1$ & $256$ & ReLu \\
\hline
$2$ & $ 128$ & ReLu \\
\hline
$3$ & $64$ & ReLu  \\
\hline
$4$ & $32$ &  ReLu \\
\hline
$5$ & $D =2$ &  Linear\\
\hline
\end{tabular}
\label{table_encoder}
\end{table}

\subsection{Channel Charting}
\begin{figure}[!t]
  \centering
  \subfloat[Siamese]{
    \includegraphics[width=0.47\linewidth]{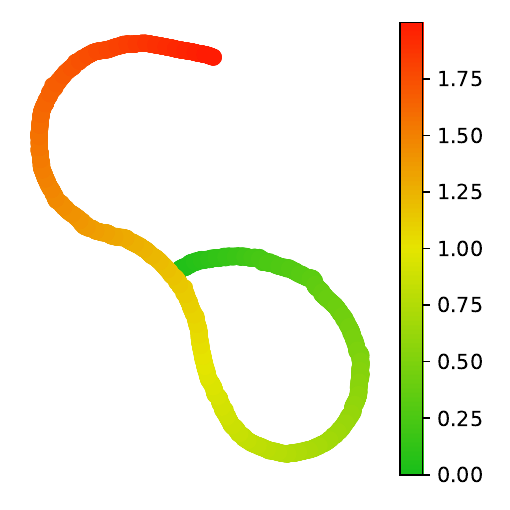}  
  } \hfill 
  \subfloat[Triplet]{
    \includegraphics[width=0.47\linewidth]{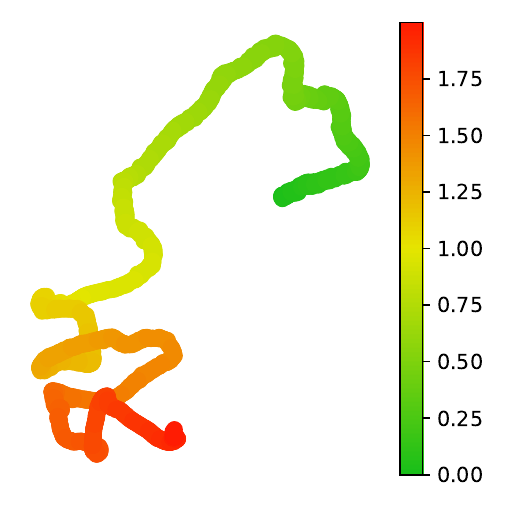}
    } \hfill
  \subfloat[PCA]{
    \includegraphics[width=0.47\linewidth]{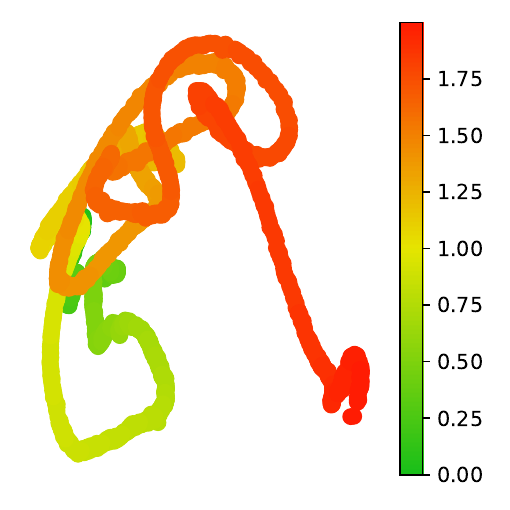}
    } \hfill
  \subfloat[AE]{
    \includegraphics[width=0.47\linewidth]{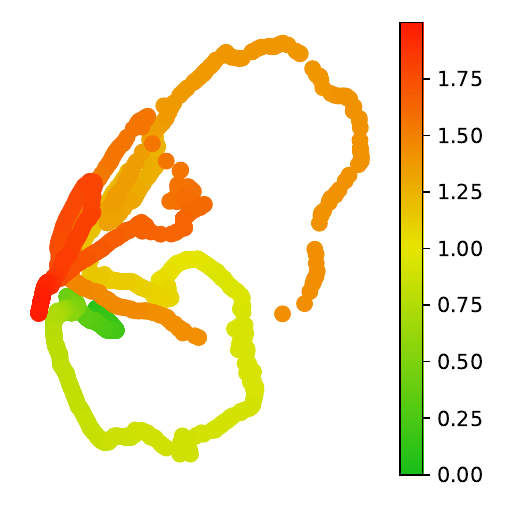}
    }
  \caption{Channel charts without perturbation. The color bar represents the UE positions along the path.}
  \label{c_charts}
\end{figure}
We present in this section the CC techniques without any perturbation. Fig. \ref{c_charts} presents the obtained CCs for PCA, AE, Triplet and Siamese network. The color bar illustrates the correspondence between the obtained chart points and the UEs positions along the path. All the results represent well the neighborhood structure in the CC, in particular, the triplet loss and the Siamese network give the best results visually. 

The neighborhood preservation metrics CT and TW are respectively shown in Tables \ref{tab_ct} and \ref{tab_tw} for different values of $K \in \{1, 50, 100, 150, 200 \}$. As the number of nearest neighbor increases, CT and TW  decrease. All the CC schemes perform very well with CT and TW close to $99\%$, especially for smallest values of $K$. The performance analysis of the privacy methods applied to the CC schemes is conducted in the following subsections.

\begin{table}[!t]
\centering
\caption{CT for different values of the $K$ nearest neighbors for the chart without perturbation.}
\renewcommand{\arraystretch}{1.5} 
\begin{tabular}{|c|c|c|c|c|c|}
\hline
$K$ &  $1$ & $50$& $100$& $150$ & $200$ \\
\hline
Triplet & $0.99$ & $0.99$ & $0.99$ & $0.98$ & $0.98$  \\
\hline
Siamese & $0.99$ & $0.99$ & $0.98$ & $0.96$ & $0.95$  \\
\hline
PCA & $0.99$ & $0.99$ & $0.99$ & $0.98$ & $0.98$  \\
\hline
AE & $0.99$ & $0.99$ & $0.99$ & $0.98$ & $0.98$ \\
\hline
\end{tabular}
\label{tab_ct}
\end{table}

\begin{table}[!t]
\centering
\caption{TW for different values of the $K$ nearest neighbors for the chart without perturbation.}
\renewcommand{\arraystretch}{1.5} 
\begin{tabular}{|c|c|c|c|c|c|}
\hline
$K$ &  $1$ & $50$& $100$& $150$ & $200$ \\
\hline
Triplet & $0.99$ & $0.99$ & $0.98$ & $0.97$ & $0.96$  \\
\hline
Siamese & $0.99$ & $0.99$ & $0.99$ & $0.98$ & $0.96$  \\
\hline
PCA & $0.99$ & $0.96$ & $0.96$ & $0.95$ & $0.95$  \\
\hline
AE & $0.99$ & $0.99$ & $0.99$ & $0.99$ & $0.98$  \\
\hline
\end{tabular}
\label{tab_tw}
\end{table}

\subsection{Differential Privacy}
CT and TW are shown for different values of $K$ and $\varepsilon$ and for all the considered channel charting methods. CT and TW are respectively presented in Tables \ref{tab_ct_dp} and \ref{tab_tw_dp} as a function of $K$ for a privacy budget $\varepsilon = 2$. Both tables show that, after perturbing the chart locations, CT and TW drop to values around $0.5$ for all the schemes. This results from the fact that the randomization of the chart coordinates by adding the Gaussian noise \emph{destroys} the structure of the charts.
\begin{table}[!t]
\centering
\caption{CT for different values of the $K$ nearest neighbors for the perturbed chart with DP}
\renewcommand{\arraystretch}{1.5} 
\begin{tabular}{|c|c|c|c|c|c|}
\hline
$K$ &  $1$ & $50$& $100$& $150$ & $200$ \\
\hline
Triplet & $0.51$ & $0.51$ & $0.51$ & $0.51$ & $0.52$  \\
\hline
Siamese & $ 0.50$ & $ 0.51$ & $ 0.51$ & $ 0.51$ & $0.52$  \\
\hline
PCA & $0.50$ & $0.50$ & $0.51$ & $0.51$ & $0.52$  \\
\hline
AE & $0.50$ & $0.50$ & $0.51$ & $0.51$ & $0.51$  \\
\hline
\end{tabular}
\label{tab_ct_dp}
\end{table}
\begin{table}[!t]
\centering
\caption{TW for different values of the $K$ nearest neighbors for the perturbed chart with DP}
\renewcommand{\arraystretch}{1.5} 
\begin{tabular}{|c|c|c|c|c|c|}
\hline
$K$ &  $1$ & $50$& $100$& $150$ & $200$ \\
\hline
Triplet & $0.51$ & $0.51$ & $0.51$ & $0.51$ & $0.52$  \\
\hline
Siamese & $ 0.50$ & $ 0.51$ & $ 0.51$ & $ 0.51$ & $0.51$  \\
\hline
PCA & $0.50$ & $0.50$ & $0.51$ & $0.51$ & $0.52$  \\
\hline
AE & $0.50$ & $0.50$ & $0.51$ & $0.51$ & $0.51$  \\
\hline
\end{tabular}
\label{tab_tw_dp}
\end{table}

CT and TW are respectively presented in Tables \ref{tab_ct_eps_dp} and \ref{tab_tw_eps_dp} for $K = 50$ and when varying $\varepsilon$ from $0.1$ to $8$. Both metrics increase when increasing $\varepsilon$ but remain close to $0.5$ even for $\varepsilon = 0.8$ where CT $= 0.54$ and TW $ = 0.53$ for the Triplet. These results indicate that, achieving good chart structure preservation after perturbation, we need much higher values of $\varepsilon$ which weaken the privacy guarantee.
\begin{table}[!t]
\centering
\caption{CT for different values of the $\varepsilon$ for the perturbed chart with DP}
\renewcommand{\arraystretch}{1.5} 
\begin{tabular}{|c|c|c|c|c|c|}
\hline
$\varepsilon$ &  $0.1$ & $2$& $4$& $6$ & $8$ \\
\hline
Triplet & $0.50$ & $0.51$ & $0.52$ & $0.53$ & $0.54$  \\
\hline
Siamese & $ 0.50$ & $ 0.51$ & $ 0.51$ & $ 0.52$ & $0.54$  \\
\hline
PCA & $0.50$ & $0.50$ & $0.51$ & $0.51$ & $0.52$  \\
\hline
AE & $0.50$ & $0.50$ & $0.51$ & $0.51$ & $0.51$  \\
\hline
\end{tabular}
\label{tab_ct_eps_dp}
\end{table}
\begin{table}[!t]
\centering
\caption{TW for different values of the $\varepsilon$ for the perturbed chart with DP}
\renewcommand{\arraystretch}{1.5} 
\begin{tabular}{|c|c|c|c|c|c|}
\hline
$\varepsilon$ &  $0.1$ & $2$& $4$& $6$ & $8$ \\
\hline
Triplet & $0.50$ & $0.51$ & $0.51$ & $0.52$ & $0.53$ \\
\hline
Siamese & $0.50$ & $0.51$ & $0.51$ & $ 0.52$ & $0.52$  \\
\hline
PCA & $0.50$ & $0.51$ & $0.51$ & $0.51$ & $0.52$  \\
\hline
AE & $0.50$ & $0.50$ & $0.50$ & $0.51$ & $0.51$  \\
\hline
\end{tabular}
\label{tab_tw_eps_dp}
\end{table}

For example, for the case of $\varepsilon = 2$, we can see in Fig. \ref{cc_perturbed_dp} how the noise destroys the intrinsics structure of the channel chart for both Siamese (Fig. \ref{cc_p_sia_dp}) and triplet (Fig. \ref{cc_p_trip_dp}) cases. 

\begin{figure}[!t]
  \centering
  \subfloat[Siamese]{ \label{cc_p_sia_dp}
    \includegraphics[width=0.47\linewidth]{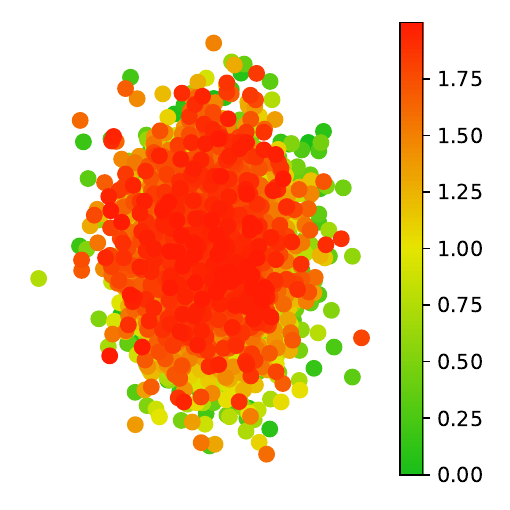}  
  } \hfill 
  \subfloat[Triplet]{\label{cc_p_trip_dp}
    \includegraphics[width=0.47\linewidth]{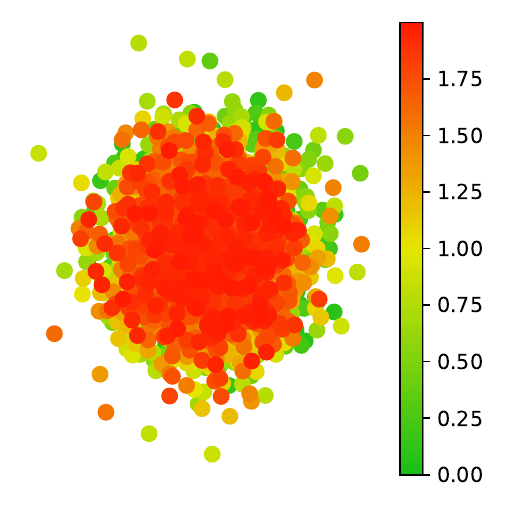}
    }
  \caption{Perturbed channel charts of the Siamese network and the triplet loss when DP is applied. $\varepsilon = 2$.}
  \label{cc_perturbed_dp}
\end{figure}

\subsection{Chart Location Indistinguishability Privacy}
We assess the performance of all the CC schemes when the chart is perturbed using the CLI framework before being released to the LBS server.

We start by showing the perturbed charts for the case of Siamese network in Fig. \ref{cc_perturbed_sia} and the Triplet-based learning in Fig. \ref{cc_perturbed_trip}. In both figures, applying planar Laplace noise based the Euclidean distance as the perturbation mechanism appears to significantly distort the CC and its neighborhood structure as in Fig. \ref{cc_p_sia_pl} for Siamese and Fig. \ref{cc_p_trip_pl} for Triplet. In contrast, the Mahalanobis norm planar Laplace mechanism preserves the neighborhood clusters after perturbation
as in Fig. \ref{cc_p_sia_mnpll} for Siamese and Fig. \ref{cc_p_trip_mnpl} for Triplet. This behavior results from the fact that in the case of MNPL, the perturbation is performed according to the local neighborhood geometry of each chart point. 

\begin{figure}[!h]
  \centering
  \subfloat[PL]{ \label{cc_p_sia_pl}
    \includegraphics[width=0.47\linewidth]{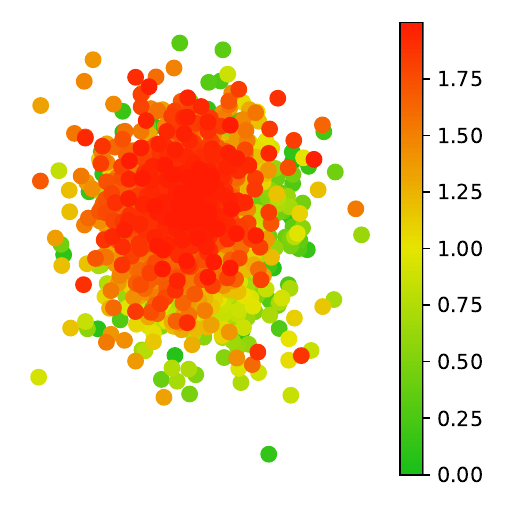}  
  } \hfill 
  \subfloat[MNPL]{\label{cc_p_sia_mnpll}
    \includegraphics[width=0.47\linewidth]{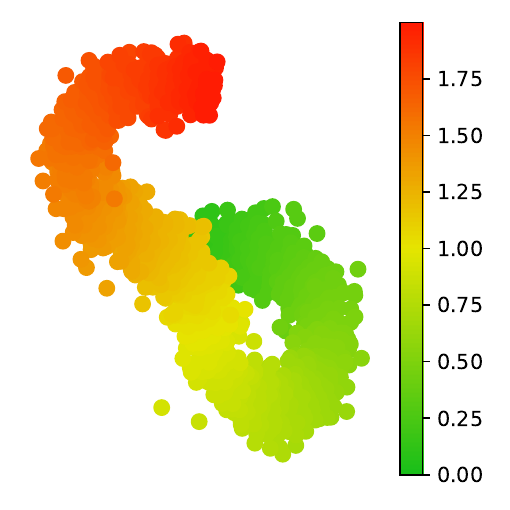}
    }
  \caption{Perturbed channel charts of the proposed Siamese network-based method for both cases of planar Laplace (PL) and Mahalanobis norm planar Laplace mechanisms. $\varepsilon = 2$.}
  \label{cc_perturbed_sia}
\end{figure}

\begin{figure}[!h]
  \centering
  \subfloat[PL]{ \label{cc_p_trip_pl}
    \includegraphics[width=0.47\linewidth]{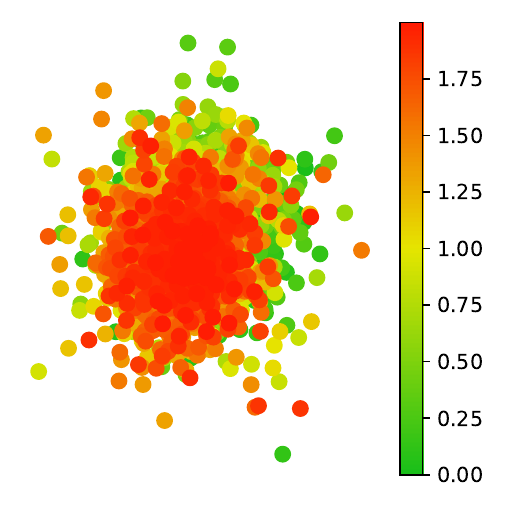}  
  } \hfill 
  \subfloat[MNPL]{ \label{cc_p_trip_mnpl}
    \includegraphics[width=0.47\linewidth]{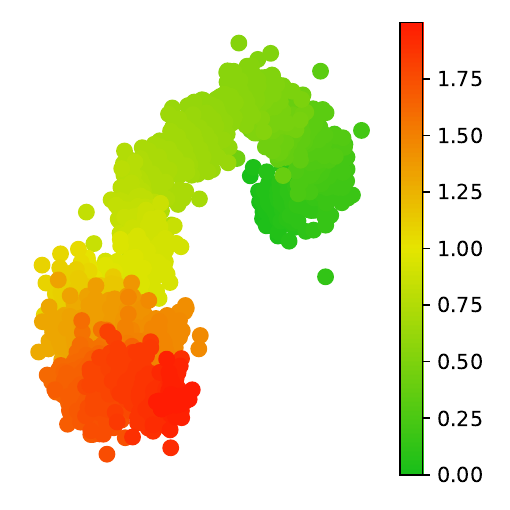}
    }
  \caption{Perturbed channel charts of the Triplet loss scheme for both cases of planar Laplace (PL) and Mahalanobis norm planar Laplace mechanisms. $\varepsilon = 2$.}
  \label{cc_perturbed_trip}
\end{figure}

In Fig. \ref{ct_tw_K}, the continuity and trustworthiness are presented as a function $K$ for a privacy budget of $\varepsilon = 2$. CT and TW slightly decrease when $K$ increases. The most significant difference is observed between the PL and MNPL perturbation mechanisms. For all the channel charting schemes, the PL perturbation gives CT and TW slightly above $0.5$, while the MNPL perturbation still give CT and TW values around $90\%$. The best results are achieved by the Siamese and Triplet methods which achieve values of $97\%$ and $94\%$, respectively, for $K = 50$. These results are consistent with the obtained channel charts in Fig. \ref{cc_perturbed_sia} and \ref{cc_perturbed_trip}. 

\begin{figure}[!t]
  \centering
  \subfloat[Continuity]{
    \includegraphics[width=0.75\linewidth]{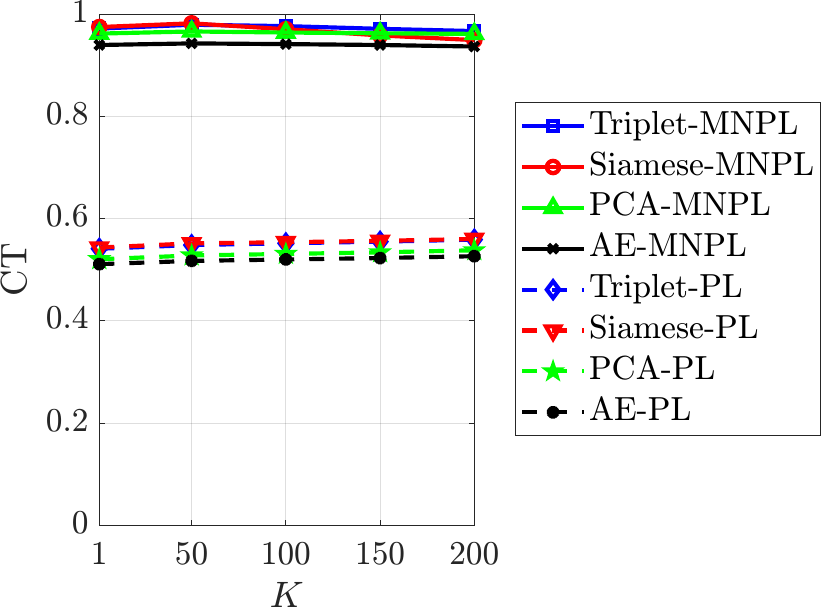}  
  } \hfill
  \subfloat[Trustworthiness]{
    \includegraphics[width=0.75\linewidth]{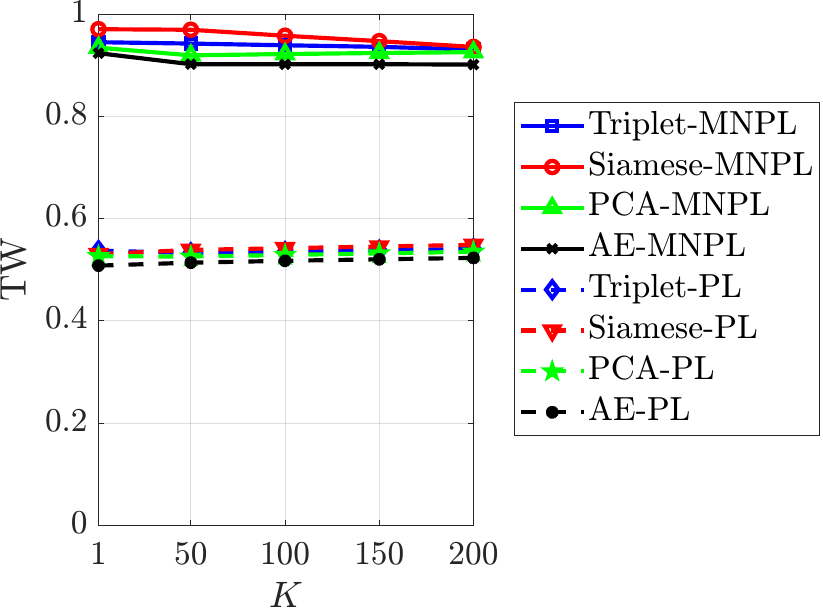}
    }
  \caption{TW and CT vs $K$ for the perturbed chart with a privacy budget of $\varepsilon = 2$.}
  \label{ct_tw_K}
\end{figure}

The assessment results of the charting performance as a function of the privacy budget $\varepsilon$ are included in Fig. \ref{ct_tw_eps} for $K = 50$. Both the continuity and trustworthiness metrics evolve with $\varepsilon$. Increasing $\varepsilon$ decreases the noise magnitude during perturbation and thereby resulting in less perturbation of the chart locations. Note that in practice the value of $\varepsilon$ is set according to the required sensitivity and privacy level of the given system.
\begin{figure}[!t]
  \centering
  \subfloat[CT]{
    \includegraphics[width=0.75\linewidth]{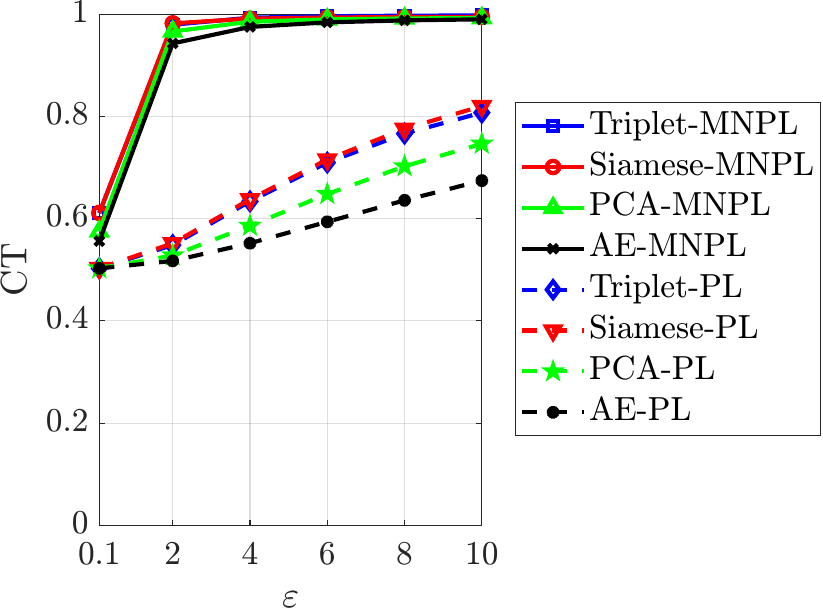} 
  } \hfill
  \subfloat[TW]{
    \includegraphics[width=0.75\linewidth]{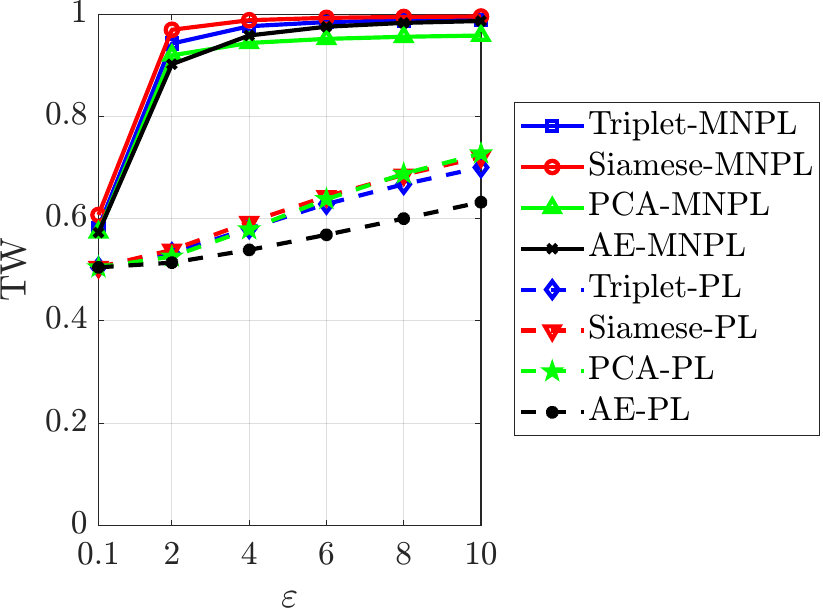}
    }
  \caption{TW and CT vs the privacy budget $\varepsilon$.}
  \label{ct_tw_eps}
\end{figure}

The impact of the privacy loss $\ell$ is evaluated in Fig. \ref{ct_tw_ell} when the privacy radius is set to $0.5$ and $K = 50$. Both continuity in Fig. \ref{ct_ell} and trustworthiness in Fig. \ref{tw_ell} increase with the privacy loss, as expected. This behavior results from the fact that the effective privacy budget increases with the privacy loss for a fixed value of the privacy radius $r$. 
\begin{figure}[!t]
  \centering
  \subfloat[CT]{\label{ct_ell}
    \includegraphics[width=0.65\linewidth]{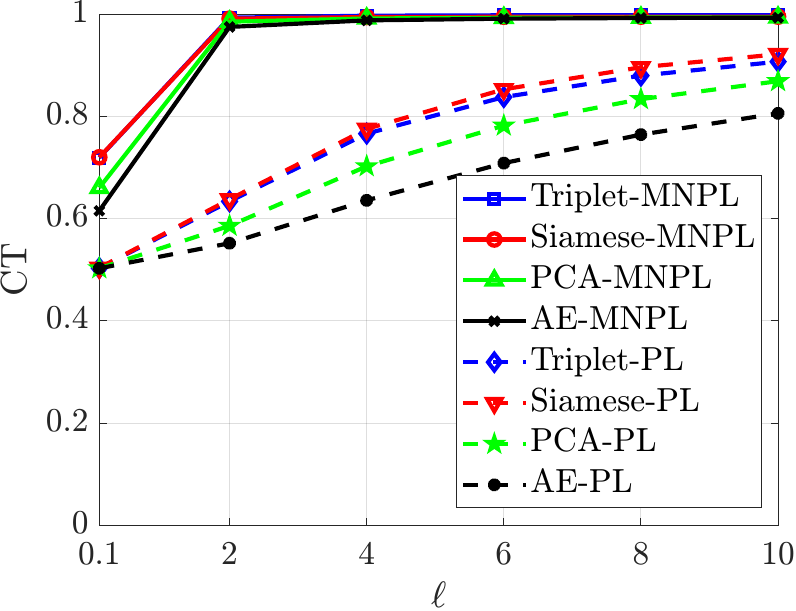}  
  } \hfill
  \subfloat[TW]{\label{tw_ell}
    \includegraphics[width=0.65\linewidth]{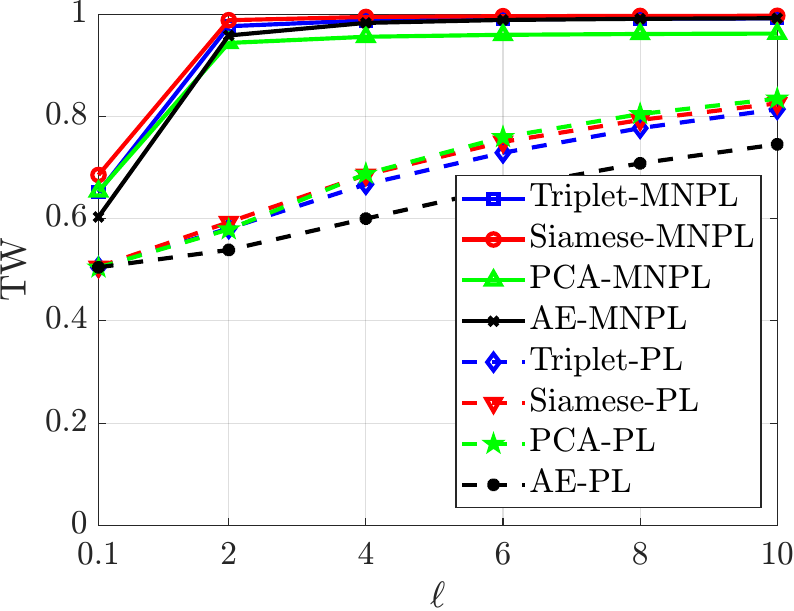}
    }
  \caption{TW and CT vs the privacy loss $\ell$. The privacy radius $r = 0.5$ and $K = 50$.}
  \label{ct_tw_ell}
\end{figure}

In contrast to $\ell$, the radius $r$ and the charting performance evolve inversely as shown in Fig. \ref{ct_tw_r}. The reason is that for a fixed value of $\ell$, the privacy budget $\varepsilon$ decreases when $r$ increases. As the privacy budget decreases, the perturbation magnitude grows, leading to a degradation of neighborhood preservation as reflected by the decrease in trustworthiness and continuity.
     
\begin{figure}[!t]
  \centering
  \subfloat[CT]{
    \includegraphics[width=0.75\linewidth]{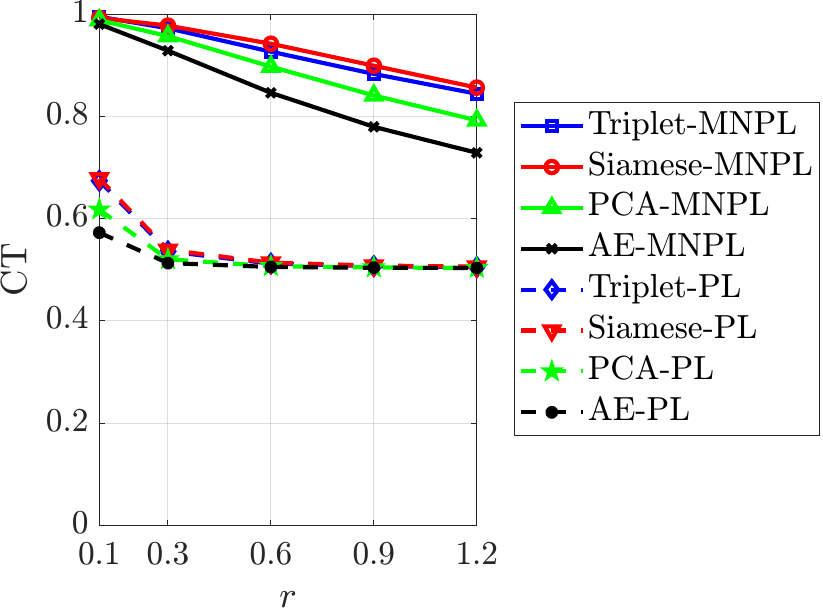}  
  } \hfill
  \subfloat[TW]{
    \includegraphics[width=0.75\linewidth]{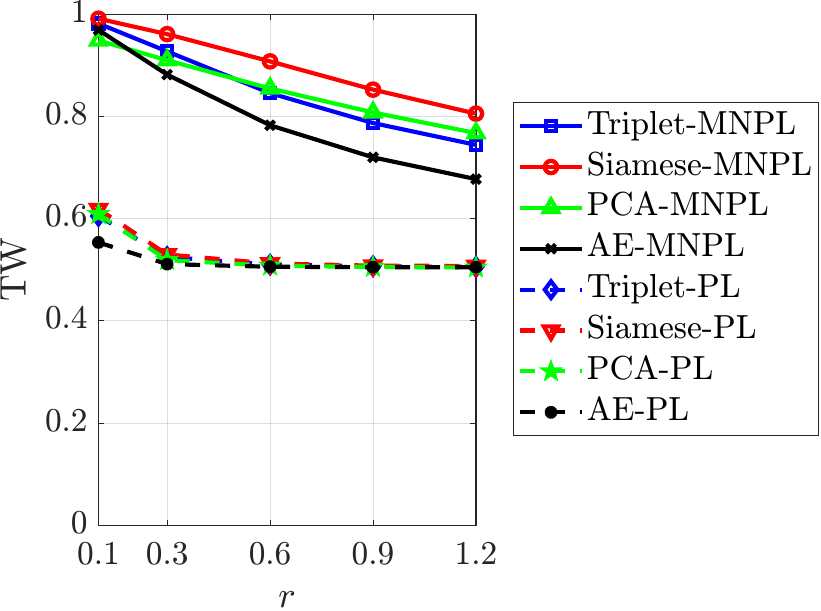}
    }
  \caption{TW and CT vs the privacy radius $r$ where the privacy loss $\ell = 0.5$.}
  \label{ct_tw_r}
\end{figure}

In Fig. \ref{rqe_eps}, Tables \ref{tab_ql_MNPL} and \ref{tab_ql_PL}, a comparative evaluation is presented between the perturbed and original charts. The range query error (RQE) in \eqref{eq_rqe} and the quality loss (QL) in \eqref{eq_ql} are considered. The range query error is illustrated in Fig. \ref{rqe_eps} for the Siamese CC scheme when varying the query radius from $0.1$ to $1.2$ and for different values of $\varepsilon \in \{0.1, 0.2, 0.3 \}$. We present only the RQE for the Siamese CC scheme, as we notice that all the considered CC methods give very similar RQE values. As the query range $r$ increases, RQE decreases for both Fig. \ref{rqe_mnlp} for the MNPL mechanism and Fig. \ref{rqe_pl} for the PL mechanism. This is the expected behavior based on the definition of RQE (Section \ref{perf_metrics}) which measures the proportion of locations that are perturbed out of the range $r$. Regarding the privacy budget $\varepsilon$, RQE decreases when $\varepsilon$ increases for a given value of $r$. Comparing both Fig. \ref{rqe_mnlp} and  Fig. \ref{rqe_pl}, MNPL achieves lower RQE values and therefore better performance than PL.
\begin{figure}[!t]
  \centering
  \subfloat[MNPL]{\label{rqe_mnlp}
    \includegraphics[width=0.65\linewidth]{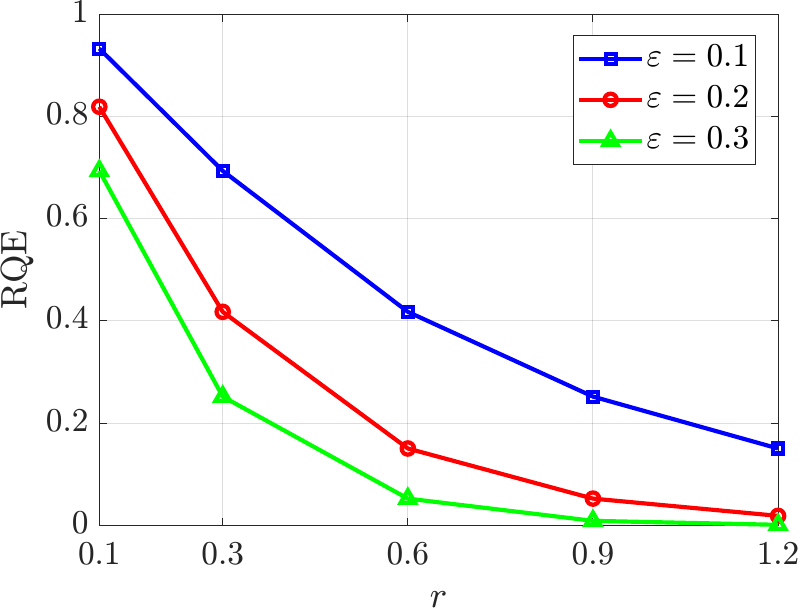}  
  } \hfill
  \subfloat[PL]{\label{rqe_pl}
    \includegraphics[width=0.65\linewidth]{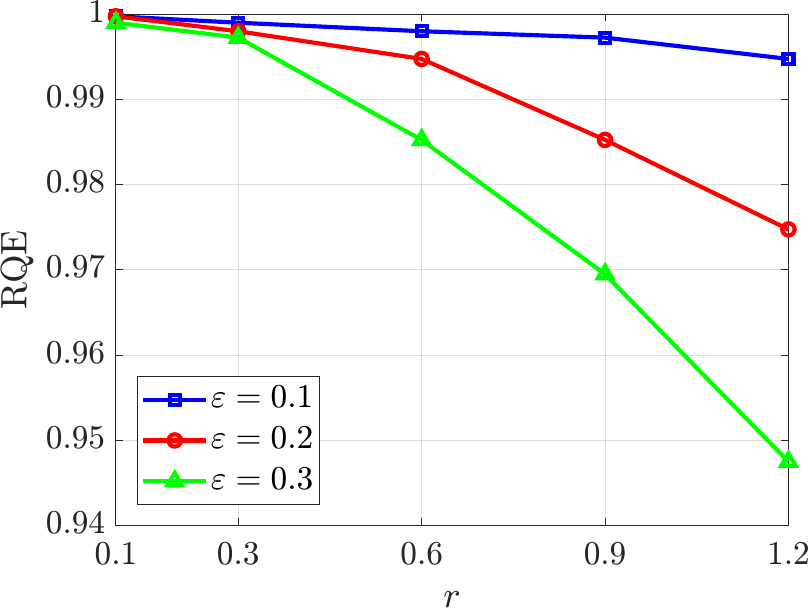}
    }
  \caption{Range query error vs the radius $r$ for different values of $\varepsilon$ for both PL and MNPL mechanisms in the case of the Siamese channel charting scheme.}
  \label{rqe_eps}
\end{figure}

Tables \ref{tab_ql_MNPL} and \ref{tab_ql_PL} include the assessment of the quality loss for MNPL and PL, respectively. QL is calculated for all the considered channel charting schemes and for different values of $\varepsilon \in \{0.1, 0.5, 1\}$. In terms of QL, all channel charting methods give very similar values for a given value of $\varepsilon$. As the privacy budget increases, the QL decreases across all channel charting schemes. The MNPL mechanism achieves the lowest QL values in all cases.
\begin{table}[!t]
\centering
\caption{Quality loss vs $\varepsilon$: case of MNPL}
\renewcommand{\arraystretch}{1.5} 
\begin{tabular}{|c|c|c|c|c|}
\hline
$\varepsilon$ &  Triplet & Siamese & PCA & AE  \\
\hline
$0.1$ & $0.66$ & $0.66$ & $0.66$  & $0.66$   \\
\hline
$0.5$ & $0.13$ & $0.13$ & $0.13$ & $0.13$  \\
\hline
$1$ & $0.07$ & $0.07$ & $0.07$ & $0.07$  \\
\hline
\end{tabular}
\label{tab_ql_MNPL}
\end{table}

\begin{table}[!t]
\centering
\caption{Quality loss vs $\varepsilon$: case of PL}
\renewcommand{\arraystretch}{1.5} 
\begin{tabular}{|c|c|c|c|c|}
\hline
$\varepsilon$ &  Triplet & Siamese & PCA & AE  \\
\hline
$0.1$ & $19.86$ & $19.86$ & $19.86$  & $19.86$   \\
\hline
$0.5$ & $3.97$ & $3.97$ & $3.97$ & $3.97$  \\
\hline
$1$ & $1.98$ & $1.98$ & $1.98$ & $1.98$  \\
\hline
\end{tabular}
\label{tab_ql_PL}
\end{table}

\section{Conclusion} \label{conclude}
{
In this paper, we have proposed a framework for location privacy in CC referred to as CLI, an extension of GI tailored to CC representations. The planar Laplace (PL) mechanism was investigated and a geometry-aware Mahalanobis norm planar Laplace (MNPL) mechanism was developed within the CLI framework. By incorporating the intrinsic geometry of the chart through a Mahalanobis distance, the MNPL mechanism perturbs chart points in a structured manner aligned with the local neighborhood structure of the chart. In addition, differential privacy using the Gaussian mechanism was considered.
Extensive evaluations across multiple CC methods demonstrated the effectiveness of the proposed framework, highlighting the trade-offs between privacy guarantees and performance. 

While the proposed CLI framework provides a first step toward geometry-aware privacy in CC, future work could investigate adaptive privacy mechanisms in which the noise distribution dynamically adjusts to variations in the spatial density of users and environmental characteristics. 
The proposed framework can also be extended to systems in which CC operates in dynamic and time-varying environments, where user trajectories and temporal correlations may reveal additional information. 

}

\section*{Acknowledgments}
{
A. K. A. Passah has been supported by SRV ENSEA, the EC through the Horizon Europe/JU SNS project ROBUST-6G (Grant Agreement no. 101139068), CNPq and the COST Action CA22168 6G-PHYSEC. A. Chorti has been partially supported by the EC through the Horizon Europe/JU SNS project ROBUST-6G (Grant Agreement no. 101139068), the IPAL Project CONNECTING, the TalCyb Chair in Cybersecurity and by the French government under the France 2030 ANR program “PEPR Networks of the Future” (ref. ANR-22-PEFT-0008 and ANR-22-PEFT-0009). R. C. de Lamare has been supported by PUC-Rio, CNPq, FAPERJ and FAPESP.}

\bibliographystyle{IEEEtran}
\bibliography{references}



\end{document}